\documentclass[11pt, oneside]{amsart}
\usepackage{amsmath}
\usepackage{amsthm}
\usepackage{amssymb} %mathrsfs}
\usepackage[mathscr]{euscript}
\usepackage{fancyhdr}
\usepackage{enumerate}
\usepackage{amscd}
\usepackage[all]{xy}
\usepackage{graphicx}
\usepackage{color}
\usepackage{hyperref}
\usepackage{comment}
\usepackage{thmtools}
\usepackage{cleveref}
\usepackage{thm-restate}
\usepackage{filecontents}
\usepackage[backend=biber]{biblatex}
\hypersetup{
    colorlinks,
    citecolor=black,
    filecolor=black,
    linkcolor=black,
    urlcolor=black
}

\usepackage{xcolor}
\usepackage{colortbl}
\usepackage{longtable}
\usepackage{tikz}
\usepackage[all]{xy}
\usepackage{graphicx}
\usepackage{biblatex}
\usepackage{tabularx}
\usepackage{slashbox}

\usetikzlibrary{backgrounds}
\definecolor{maroon}{cmyk}{0, 0.87, 0.68, 0.32}
\definecolor{Dandelion}{cmyk} {0 0.29 0.84 0} % PANTONE 123
\definecolor{Apricot}{cmyk} {0, 0.32, 0.52, 0} % PANTONE 1565
\definecolor{BurntOrange}{cmyk}{0, 0.51, 1, 0} % PANTONE 388
\definecolor{Bittersweet}{cmyk}{0, 0.75, 1, 0.24} % PANTONE 167
\definecolor{RedOrange}{cmyk}{0, 0.77, 0.87, 0} % PANTONE 179
\definecolor{Mahogany}{cmyk}{0 0.85 0.87 0.35} % PANTONE 484
\definecolor{Magenta}{cmyk}{0, 1, 0, 0} % PANTONE PROCESS-MAGENTA
\definecolor{DarkOrchid}{cmyk}{.40, 0.80, 0.20, 0} %  No PANTONE match
\definecolor{Brown}{cmyk} {0 0.81 1. 0.60} % PANTONE 1615
\definecolor{Purple}{cmyk}{0.65 0.86 0 0} % PANTONE PURPLE
\definecolor{Plum}{cmyk}{0.50, 1., 0, 0} % PANTONE 518
\definecolor{MidnightBlue}{cmyk}{0.98, 0.13, 0, 0.43} % PANTONE 302
\definecolor{NavyBlue}{cmyk}{0.94, 0.54, 0, 0} % PANTONE 293
\definecolor{RoyalBlue}{cmyk}{1. 0.50 0 0} % No PANTONE match
\definecolor{Blue}{cmyk}{1. 1. 0 0} % PANTONE BLUE-072
\definecolor{Cerulean}{cmyk}{0.94 0.11 0 0} % PANTONE 3005
\definecolor{ForestGreen}{cmyk}{0.91 0 0.88 0.12} % PANTONE 349
\definecolor{PineGreen}{cmyk}{0.92, 0, 0.59, 0.25} % PANTONE 323
\definecolor{OliveGreen}{cmyk}{0.64, 0, 0.95, 0.40} % PANTONE 582
\definecolor{tablegreen}{RGB}{124, 179, 113}
\definecolor{GovernorBlue}{RGB}{54, 54, 180}
\newtheorem*{theorem*}{Theorem}

\theoremstyle{plain}
\newtheorem{lemma}{Lemma}[section]

\newtheorem{proposition}[lemma]{Proposition}
\newtheorem{theorem}[lemma]{Theorem}

\newtheorem{corollary}[lemma]{Corollary}

\numberwithin{equation}{section}

\theoremstyle{remark}

\theoremstyle{definition}
\newtheorem{remark}[lemma]{Remark}
\newtheorem{example}[lemma]{Example}
\newtheorem{definition}[lemma]{Definition}

\usepackage[tmargin=1in,bmargin=1in,lmargin=1in,rmargin=1in]{geometry}
\begin{document}
\title{Basic Strategy for Some Simplified Blackjack Variants}
\author{Basmalah Asad, Daniel Martin}
%\affiliation{Mathematics Department, University of Hartford.}

\date{\today}

\begin{abstract}
In this paper, we calculate a basic strategy for several variations of a simplified version of Blackjack.  In short, for these variants the player has only the two options of hit or stand, and they may only make either decision once. Other minor variations including rule modifications, changes in number of decks, alternate payout structures, and different given information about dealer's hand are also considered. An interesting theoretical result regarding the asymptotic behavior of the basic strategy and overall expectation as the number of decks increase is also proved. 
\end{abstract}

\maketitle

\section{Introduction}
\label{sec:introduction}
Blackjack is one of the most popular games in a casino due to its simplicity and the ability of the player to improve their winnings by learning basic strategy - a term used to refer to the collection of optimal decisions made in each betting scenario. The classical resource for the study of Blackjack is E.O. Thorp's 1966 book \textit{Beat the Dealer} \cite{Thorp1966}. In this book, Thorp gives a strategy for beating the house edge called ``counting cards", where the player tracks characteristics of the deck composition as play continues and uses this information to increase bet sizing in favorable situations. Combined with basic strategy, card counting gives the player a positive expected value, while basic strategy alone yields a small negative expectation of around $-0.005$ for a \$1 bet. From an academic point of view, the formulation of basic strategy for Blackjack is typically done through simulating a very large number of hands, and is formidable to calculate exactly by-hand - even a single hand layout.  In more recent years, J. Lin et al. in \cite{Lin2019} and Y. Mura and H. Wada in \cite{Mura2020} have explored how quantum computing can be used to develop a strategy that yields a higher expectation than the classical basic strategy.  In \cite{Baldwin1956}, R.R. Baldwin et al. attempt to give a mathematical proof for the basic strategy formulation, but ultimately make simplifying assumptions and fall short of full generality - calling their work a  ```practical" optimum strategy'. Although obtaining these results in full generality has not been done, there has been progress in calculating the exact expected values for splitting pairs in \cite{Nairn2019} by J. Nairn. One of the main reasons for the computational difficulty is the ability of both the player and dealer to make a sequence of actions.  The resulting calculations require several rounds of conditional probabilities.  This issue is greatly simplified in our variations, making the calculations more manageable by-hand.

Other authors have gone down similar paths as we have by studying simplified variants. In the book \emph{The Theory of Gambling and Statistical Logic} \cite{Epstein1967}, R.A. Epstein introduces a simplified blackjack variant, ``Grayjack", which uses a 13-card deck. The main goal of the game is to achieve a hand total of 13 instead of 21. In \cite{Ethier2019}, S.N. Ethier and J. Lee study a further simplified game ``Snackjack" by removing all cards except two aces, two twos, and four threes with the goal being to reach a total of 7. They also have modified the corresponding rules of the dealer. In addition to introducing an interesting toy model of the game, Ethier and Lee provide a basic strategy, as well as a system for card counting when multiple decks are introduced. Despite the development of their basic strategy in both cases, most of the by-hand calculations are a similar level of difficulty to that of traditional Blackjack. On the other hand, although our own derivation of basic strategy is ultimately done by a computer, it is not done through simulation and could theoretically be done by hand - albeit impractically. 

Our motivation behind studying these variations was to provide a setting in which a precise exploration of a simplified version of the game can be done, while maintaining the feel of the original game.  The derivation of optimal strategy is relatively straight-forward, can be done without approximation, and is computationally less demanding than traditional Blackjack.  Briefly, one needs to calculate the probability of winning, losing, and pushing (tying) for each decision and use these to calculate the expectation for each. Once we have these, a comparison of the expected values for each decision determines the correct action to take. We demonstrate a concrete example of this procedure in Example \ref{ex:stage1} and generally in Section \ref{sec:deckcomp}.

We have collected many interesting computational results in Section \ref{sec:computations} about several variations of the game, where we give the overall expected value and probabilities when we play according to the basic strategy. In Appendix \ref{appendix:tables}, we provide several tables that outline this strategy for these variations. Given these variations, most give the player a positive overall expectation, but we do find some that yield a negative expectation. In fact, we are able to provide a variation of our game that gives effectively the same expectation as traditional Blackjack.

Upon exploring the changing of deck composition - particularly increasing the number of decks - we were confronted with an interesting theoretical question. In particular, we first sought to study what happens to the basic strategy and overall expectation when we play the game with more decks. Upon revising our program to do this, we found that the overall expectations do go down when more decks are played with, but not by a significant amount. However, we wondered what the asymptotic behavior of these results tend to as the number of decks increase. The answer itself is not too far from the results we obtained playing with one deck, but led us to an interesting conjecture. Is playing the game with an infinite number of decks effectively the same as playing the game with one deck, but \textit{with replacement}? This conjecture is investigated in Section \ref{sec:deckcomp} and is proved in the affirmative.  We give the simplified statement here:

\begin{theorem*}\label{simplifiedmain}
	 As the number of decks approaches infinity, the basic strategy and overall expectation when playing our Blackjack variants with $n$ decks converge to the basic strategy and overall expectation when playing using one deck with replacement.
\end{theorem*}

It was not our original intention to glean any new information about traditional Blackjack, but given our proof, it seems reasonable to conjecture that this result is also true in the traditional setting. Interestingly, if it is true in the traditional Blackjack case, finding a practical way to play using one deck with replacement would completely remove the ability of the player to count cards. That is, since the deck composition is unchanged when playing with replacement, counting cards would yield no advantage. If the overall expectation of the player goes down, as it does in our variations, it would also provide the house with a larger edge.

%%%%%%%%%%%%%%  ACKNOWLEDGEMENTS

\section*{Acknowledgments}
Support for this work was provided by The Women’s Advancement Initiative’s Dorothy Goodwin Scholars Program, which was made possible thanks to a generous bequest from Dorothy Goodwin. The Women’s Advancement Initiative at the University of Hartford is proud to continue the legacy of advancing each woman’s potential in the Hartford College for Women tradition. This support does not necessarily imply endorsement by the University of Hartford or The Women’s Advancement Initiative of research conclusions.

%%%%%%%%%%%%%% SECTION 2   

\section{Background}
\label{sec:background}
In this section, we clearly outline the rules of our simplified game, as well as how these rules change when we consider variations. To start, we remind the reader of the rules that are common between our variations and the traditional game, as well as highlight the main differences. First, each card in the deck has a rank, and it is this that determines the value of a card. More specifically, cards with numerical ranks (2 - 10) are worth their face value, face cards are valued at 10, and aces are worth the higher of 1 or 11 so long as the hand value is below 22. For every variation we consider, it is assumed that the dealer is playing against only one player. To initiate the game, the player places a bet, and two cards are dealt to the player and dealer, both of which are visible. In traditional Blackjack, only one of the two dealer cards is showing.

The key difference between our variations and the traditional game is that there are only two possible actions the player and dealer can take, and they may only take this action \textit{once}. The only actions are to take an additional card - \textit{hit} - or not take another another card - \textit{stand}. In traditional Blackjack, the player also has the options to \textit{double-down} and \textit{split} pairs. To double down the player doubles their bet size, but only receives one additional card, and if the player is dealt a pair of cards the same rank, they can split into two separate hands and bet on each. We have chosen to remove these options from our game for simplicity, but we could modify our variations to include it if desired. 

Assuming that both the player and dealer are not initially dealt hands valued at $21$, the player acts first with their decision to hit or stand. If the player's hand value exceeds 21, or \textit{busts}, the player loses the game and, consequently, the amount they bet without the dealer getting the opportunity to play their hand. Otherwise, if the player does not bust, the dealer then takes their action, which is completely determined by the following: hit on any hand whose value is less than 17 or on any soft 17, and stand on any other hand. We call a hand \textit{soft} when it contains an ace and \textit{hard} when it does not. If the dealer busts, then the player wins, and if neither bust, the winner is determined by who has the larger hand value. If the player wins in either case, they get their bet back and profit the size of their bet.

Some subtlety occurs in situations where the player or dealer is initially dealt a hand valued at $21$.  If the player is dealt a hand valued at $21$, two outcomes are possible. The hand is declared a \textit{push} (tie) if the dealer, upon discreetly checking their hand value, has a hand that is also valued at $21$.  If the dealer does not have a hand valued at $21$, the player wins $1.5$-times their bet, and the dealer does not have the opportunity to play their hand.  In this case, the player is said to have a \textit{natural}, or \textit{Blackjack}.  On the other hand, if only the dealer's hand is valued at $21$, the player does not get the opportunity to play their hand, and thus loses their bet.  \textit{Insurance} is an option for the player when the dealer is showing an ace to protect themselves against the situation that the dealer has Blackjack, but it is not offered in our version of the game.

In summary, the player can win in three ways:  player has a natural $21$ while the dealer does not, player has a hand total larger than the dealer's hand total while less than $22$, or the dealer busts while the player has a hand total less than $22$.  The player pushes when their total is less than $22$ and equal to the dealer's hand total, and loses in all other situations.

\subsection{Variations}
\label{sec:background-variations}

Next, we describe the variations that we will consider. There are three variations that produce the most significant changes in results and they pertain to the visibility of the dealer's starting hand. Within each of these three variations, we will also make smaller modifications to the rules, payouts, and deck-composition. 

In traditional Blackjack, only one of the two dealer's cards is visible. We refer to the visible card as the \textit{up-card}, and the non-visible card as the \textit{down-card}. Consequently, we refer to this as the \textit{one up-card} variation since only one of the dealer's cards is visible.  The other variations we will consider are \textit{two up-card} and \textit{no up-card}, where both dealer's cards are visible and neither dealer cards are visible, respectively.

Within each of the three variations above, we will also modify the deck composition. Modern traditional Blackjack uses several decks - up to 8 - when played at a casino, but our main variation will only be played with one deck.  However, we will analyze how the expected value changes as the number of decks that are being played with increases.  We refer to a variation as $n$-deck, when we use $n$ decks within each variation.

We will also consider modifications the rules - particularly for the dealer. In the variation we described above, the dealer hits on all hands valued less than 17, or on any soft 17. We will discuss how bending these rules changes the overall expectation and basic strategy, as well as which rule is most optimal for the dealer. 

Lastly, it is becoming more popular for casinos to offer what is referred to as ``6 to 5" Blackjack. This means that instead of winning \$3 for every \$2 bet, or 1.5-times the player's bet size, in the case of a natural, the player will only earn \$6 for every \$5 bet, or 1.2-times the player's bet size. As we will see in Section \ref{sec:computations}, this will significantly change the overall expectation of the player.

%%%%%%%%%%%%%%%%%%%%%% SECTION 3

\section{Computational Results For Variations} % Section 4
\label{sec:computations}
In this section of the paper, we give our main computational results regarding the basic strategy, overall expectation, and probabilities of winning, losing and pushing for several variations of the game. In particular, we look at how these quantities change when we vary the known information about the dealer's hand, using different payout structures for naturals, and change rules regarding when the dealer hits.  In short, we observe that the expectation goes down as the amount of known information about the dealer's hand goes down, but they all remain positive. Additionally, the expectation goes down, naturally, when the compensation for a natural is modified and in fact produces a variation that yields a negative expectation. Lastly, we find the optimal rule that the dealer must play by, in terms of when to hit, that will provide the player with the smallest edge.

Before an overall computational analysis of these variations can be done, the basic strategy for each must be developed. In short, our approach to developing the basic strategy is to calculate the expected value when the optimal decision is made for each permutation of starting values for the player and dealer. To do this, we calculate the probability of winning, losing, and pushing when the player always hits, as well as when the player always stands. Given these probabilities, we compute the expected value for each decision and base the optimal decision on whichever expected value is larger. A complete example involving a specific hand layout for the player and dealer is shown in Appendix \ref{sec:derivation}, which will help elucidate the abstract computations in our proofs in Section \ref{sec:deckcomp}.  

In all of the variations in this section, we remind the reader that there is only one player and the dealer, and only one deck in use. The analysis of the dealer information variations when more decks are introduced in Section \ref{sec:deckcomp}.

\subsection{Dealer Information Variations}

In this subsection, we give the overall probabilities of winning, losing, and pushing, along with the expectation when we play according to the basic strategy for each of the three dealer information variations: two up-card, one up-card, and no up-card. We refer the reader to Appendix \ref{appendix:tables}, where we provide the basic strategy tables for each variation. These computational results are summarized in the following theorem:

\begin{theorem}\label{theorem:deckcomp}
	The overall probabilities of the player winning, losing, and pushing and expectation when each variation is played using one deck are summarized as follows: 
	\begin{table}[h!]
		\setlength{\arrayrulewidth}{0.5mm}
		\centering
		\begin{tabular}{|c|c|c|c|}
			\hline
			\cellcolor{RoyalBlue!20}\emph{\textbf{Variation:}}  & \cellcolor{RoyalBlue!20}\emph{\textbf{Two up-card}} & \cellcolor{RoyalBlue!20}\emph{\textbf{One up-card}} & \cellcolor{RoyalBlue!20}\emph{\textbf{No up-card}} \\ \hline
			\cellcolor{RoyalBlue!20}$P(W)$     & $0.470912 $   & $0.455694   $ & $0.454128   $\\ \hline
			\cellcolor{RoyalBlue!20}$P(L)$     & $0.446386 $   & $0.464616  $  & $0.468471  $ \\ \hline
			\cellcolor{RoyalBlue!20}$P(T)  $   & $0.082702 $   & $0.0797   $   & $0.077401 $  \\ \hline
			\cellcolor{RoyalBlue!20}$E{[}X{]} $& $0.047771$   & $0.014324$    & $0.008902$  \\ \hline
		\end{tabular}
	\end{table}\\
	When playing according to the basic strategy, the events $W,\,L,\,T$ are of the player winning, losing, and pushing, and the random variable $X$ is the payout when the player makes a $\$1$ bet.  
	\end{theorem}

Given the above results, we take a moment to compare across the variations. Our first interesting result is that they all yield positive expectation for the player.  That is, given perfect decision-making, the player can beat the game long-term. Naturally, we see that when more information about the dealer's hand is known, the expectation is at its largest, and as we lose information the expectation goes down. It is also interesting to note that only in the two up-card variation, we are actually more likely to win than we are to lose. Despite this, the expectation is still positive in all due to expected value contributions coming from when the player is dealt a natural. We also see that all of the above variations have a higher expectation than traditional Blackjack, which has a negative expectation when played optimally and without counting cards.

\subsection{Payout Variations - ``6 to 5" Blackjack}

In this subsection, we collect the computational results regarding the variations involving alternate payout structures.  In particular, the following theorem gives the overall results when naturals are paid out ``6 to 5" instead of ``3 to 2". That is, the player now only earns 1.2x their bet rather than 1.5x. It is becoming increasingly more common for casinos to offer this particular variation of Blackjack.  Because this variation only affects the payouts for naturals, where no decision is to be made, the basic strategy tables for this variation are the same as in the 3 to 2 blackjack cases. We refer to the tables in Appendix \ref{appendix:tables}.

\begin{theorem}\label{theorem:65}
	The following table summarizes the overall probabilities and expectations when each variation is played using one deck, and also assuming a ``6 to 5" payout structure for naturals:
	\begin{table}[h!]
		\setlength{\arrayrulewidth}{0.5mm}
		\begin{tabular}{|c|c|c|c|}
			\hline
			\cellcolor{RoyalBlue!20}\emph{\textbf{Variation:}} & \cellcolor{RoyalBlue!20}\emph{\textbf{Two up-card}} & \cellcolor{RoyalBlue!20}\emph{\textbf{One up-card}} & \cellcolor{RoyalBlue!20}\emph{\textbf{No up-card}} \\ \hline
			\cellcolor{RoyalBlue!20}$P(W)$  & $0.470912$   & $0.455694$  & $0.454128$   \\ \hline
			\cellcolor{RoyalBlue!20}$P(L)$  & $0.446386$   & $0.464616$  & $0.468472$    \\ \hline
			\cellcolor{RoyalBlue!20}$P(T)$  & $0.082702$   & $0.07797$   & $0.0774$ \\ \hline
			\cellcolor{RoyalBlue!20}$E[X]$  & $0.033824$   & $0.003761$  & $-0.005046$  \\ \hline
		\end{tabular}
	\end{table}\\
	When playing according to the basic strategy, the events $W,\,L,\,T$ are of the player winning, losing, and pushing and the random variable $X$ is the payout when the player makes a $\$1$ bet.  
\end{theorem}

When we compare these expectations to our results from Theorem \ref{theorem:deckcomp}, we see that this payout modification yields noticeable changes in expectation to each of our three variations. This is the first instance where the expectation for one of our variations becomes negative - in the no up-card variation. In fact, this expectation is extremely close - equal to the thousandths decimal place - to the that of traditional Blackjack.

\subsection{Dealer Rule Changes}

For the last of our variations, we will modify the rules the dealer plays by when it comes to when they hit. In particular, our interest is in how the expectation changes as a function of these rule changes. We remind the reader that in traditional Blackjack and the variations we have studied in this work so far, the dealer typically hits on any hand valued less than 17 and on any soft 17. This rule is commonly referred to as \textit{H17}. When the dealer hits on any hand valued less than 17 and stands on soft 17, it is referred to as \textit{S17}. In what follows, we will be changing the value that the dealer is required to hit on and examining how this changes the overall expectation.

\begin{theorem}\label{theorem:HS}
	The entries of the following table give the overall expectation in each variation across several restrictions on when the dealer must hit:
	\begin{table}[h!]
		\setlength{\arrayrulewidth}{0.5mm}
		\centering
		\begin{tabular}{|c|c|c|c|c|c|c|c|c|}
			\hline
			\cellcolor{RoyalBlue!20}\emph{\textbf{Dealer Rule:} }& \cellcolor{RoyalBlue!20}\emph{\textbf{S15}}   & \cellcolor{RoyalBlue!20}\emph{\textbf{H15}}   & \cellcolor{RoyalBlue!20}\emph{\textbf{S16}}  & \cellcolor{RoyalBlue!20}\emph{\textbf{H16}}   & \cellcolor{RoyalBlue!20}\emph{\textbf{S17}}    & \cellcolor{RoyalBlue!20}\emph{\textbf{H17}}   & \cellcolor{RoyalBlue!20}\emph{\textbf{S18}}    & \cellcolor{RoyalBlue!20}\emph{\textbf{H18}}   \\ \hline
			\cellcolor{RoyalBlue!20}\emph{\textbf{\textbf{Two up-card}}} &$  0.05718  $&$  0.05337 $ &$  0.04428  $&$  0.04217 $ & $ 0.04858 $ &$  0.04777  $& $ 0.07095 $ & $ 0.07202 $ \\ \hline
			\cellcolor{RoyalBlue!20}\emph{\textbf{One up-card}} & $ 0.03038 $  & $ 0.02663 $ &$  0.01523 $ & $ 0.01365  $& $ 0.01421 $ &$  0.01432 $ & $ 0.03493 $ & $ 0.03724 $ \\ \hline
			\cellcolor{RoyalBlue!20}\emph{\textbf{No up-card} } &$  0.02504 $ & $ 0.0213 $ & $ 0.00816 $ & $ 0.0061  $& $ 0.00871 $ &$  0.0089  $& $ 0.03183 $ &$  0.03413 $ \\ \hline
		\end{tabular}
	\end{table}\\
In the above table, H refers to the rule that the dealer must hit on soft hands valued at the given number, and S indicates that they must stand on soft hands valued at that number. In either case, they must hit on all hands valued less than the given number and stand on all valued at and above.
\end{theorem}

As we can see from the table in Theorem \ref{theorem:HS}, the player has the lowest expectation in all variations when the dealer must hit on any hand valued less than 16 and on any soft 16. All expectations to the right and left of this column increase, which means this is where the dealer is playing its own optimal strategy. Our program also indicates that the expectation goes all the way up to 0.273143 when the dealer always hits, and 0.289466 when the dealer always stands in the no up-card variation, for example. Note that the column for S17 is exactly as is given in the table in Theorem \ref{theorem:deckcomp}.

%%%%%%%%%%%%%%%%%%%%%%%% SECTION 4

\section{Theoretical Results Regarding Deck Composition Variations} % Section 5
\label{sec:deckcomp}

In Section \ref{sec:computations}, we have considered variations involving changes to the known information about the dealer's starting hand, but all of these variations assumed that the game is being played with one deck. We discuss here the asymptotic behavior of the basic strategy, overall expectation, and probabilities of winning, pushing, and losing as the number of decks increases. In Subsection \ref{sec:longtermbehavior} below, we prove that the asymptotic behavior of the two up-card variation as the number of decks increases is that of using one deck, but with replacement. In Subsection \ref{sec:1dwr}, we summarize the overall probabilities and expectation from modifying our program in each of the variations to play using one deck with replacement.

\subsection{Asymptotic Behavior of the Two Up-card Variation}
\label{sec:longtermbehavior}
Our first realization is that the optimal strategy and overall expectation does not appear to change as significantly as we expected with the addition of more decks. With this in mind, we shifted our focus to studying the asymptotic behavior of these results. Our main results in this section prove the convergence of the limits of the overall probabilities and expectations from using $n$ decks as the number of decks approaches infinity - which we will see is effectively the same as playing the game using one deck, but \textit{with replacement}. 

Our approach to proving this will be similar to the path outlined in Section \ref{sec:derivation}. That is, we first calculate the probabilities and expectations when the initial \textit{hand layouts} are fixed as a function of the number of decks, $n$. Upon establishing these, we then calculate the analogous quantities when only the initial \textit{hand values} are fixed. This will allow us to generate the basic strategy for both variations. Lastly, we compute the overall probabilities and expectations when we play according to this basic strategy.

For definiteness, all of the forthcoming results are proved in the setting of the two up-card variation. However, the outline of the proof translates directly to the other two deck composition variations without much effort.  In fact, it would seem a reasonable conjecture that these limiting arguments would carry over to traditional Blackjack. 

We begin by defining the relevant probability measures and expectations subject to the different conditions that are known, as well as the random variables that correspond to the payouts.

\begin{definition}\label{def:measures}

Our experiment will be dealing two cards at random to the player and dealer from either a collection of $n$ decks or from 1 deck with replacement. We define the following conditional probability measures:
\begin{itemize}
	\item  $P''( \,\,\, \cdot \,\,\, ) = P(\,\,\,\cdot\,\,\, | \text{ initial hand layouts })$,  
	\item  $P'(\,\,\, \cdot\,\,\, ) = P(\,\,\,  \cdot\,\,\,  | \text{ initial hand values })$.
	\end{itemize}

Each of the conditional probability measures has a corresponding conditional expectation:
\begin{itemize}
	\item  $E''[ \,\,\, \cdot \,\,\, ] = E[\,\,\,\cdot\,\,\, | \text{ initial hand layouts }] $,  
	\item  $E'[\,\,\, \cdot\,\,\, ] = E[\,\,\,  \cdot\,\,\,  | \text{ initial hand values }]$.
\end{itemize}

	Additionally, we define the random variables that give the payout when we hit or stand:
\begin{itemize}
	\item  $X_{s,n}$: Payout when the player stands when playing using $n$ decks,
	\item  $X_{h,n}$: Payout when the player hits when playing using $n$ decks,
	\item  $X_{s,\infty}$:  Payout when the player stands when playing using 1 deck with replacement,
	\item  $X_{h,\infty}$:  Payout when the player hits when playing using 1 deck with replacement. 
	\end{itemize}

\end{definition}

We note that the prime notation in the above definition does \textit{not} indicate a derivative.

\begin{proposition}
\label{proposition:handlayout}
Suppose that the initial hand layouts for the player and dealer are fixed while playing the two up-card variation, and let $n$ be the number of decks. Then, as $n \to \infty$, we have the following:
\begin{itemize}
	\item $P''(W_{s,n}) \to P''(W_{s,\infty})$ and $P''(W_{h,n}) \to P''(W_{h,\infty})$
	\item $P''(T_{s,n}) \to P''(T_{s,\infty})$ and $P''(T_{h,n}) \to P''(T_{h,\infty})$
	\item $P''(L_{s,n}) \to P''(L_{s,\infty})$ and $P''(L_{h,n}) \to P''(L_{h,\infty})$
	\item $E''[X_{s,n}] \to E''[X_{s,\infty}]$ and $E''[X_{h,n}] \to E''[X_{h,\infty}]$.
\end{itemize}
where $W_{s,n}, \, T_{s,n}, \, L_{s,n}$ and $W_{s,\infty}, \, T_{s,\infty}, \, L_{s,\infty}$ are the events of the player winning, pushing, and losing when the player stands using $n$ decks and $1$ deck with replacement, respectively. Similarly, the subscript $h$ indicates the decision to hit.
\end{proposition}

We re-iterate that all quantities in the above Proposition \ref{proposition:handlayout} are conditioned on the initial player and dealer hand layouts being known, and are defined in Definition \ref{def:measures}. We refer the reader to Appendix \ref{sec:derivation}, where they are calculated in a concrete example. 

\begin{proof}
	Given fixed initial hand layouts for the player and the dealer, assume first that the player chooses to stand. Then, when we use $n$ decks, 
	\begin{align*}
		P''(W_{s,n}) &= \dfrac{a\cdot 4n - w}{52n - 4} & a &\in  \{0,\,1,\,\dots,\,13\} & w &\in \{0,\,1,\,2,\,3,\,4\}, \\
		P''(T_{s,n}) &= \dfrac{b\cdot 4n - t}{52n - 4} & b &\in  \{0,\,1,\,\dots,\,13 \} & t &\in \{0,\,1,\,2,\,3,\,4\}, \\
		P''(L_{s,n}) &= \dfrac{c\cdot 4n - \ell}{52n - 4} & c &\in \{0,\,1,\,\dots,\,13\} & \ell &\in \{ 0,\,1,\,2,\,3,\,4\},
	\end{align*}
	where $a,\,b,\,c$ count the number of ranks that result in the player winning, pushing, and losing, respectively, and consequently, $a + b + c = 13$. Note that these $a,\,b,\,c$ are dependent on the initial hand layouts. The variable $w$ counts the number of cards that are no longer in the deck that would help the player win. That is, cards that help the player win that are already in the initial hands of the player or dealer. Similarly, $t$ and $\ell$ count the number of cards no longer in the deck that would help the player push and lose, respectively. Consequently, $w + t + \ell = 4$ since there are four cards removed from the deck - two for each player and dealer. 
	
	When we use 1 deck with replacement, we have
	\begin{align*}
		P''(W_{s,\infty}) &= \dfrac{a\cdot 4}{52} & a &\in \{ 0,\,1,\,\dots,\,13\}, \\
		P''(T_{s,\infty}) &= \dfrac{b\cdot 4}{52} & b &\in \{ 0,\,1,\,\dots,\,13\}, \\
		P''(L_{s,\infty}) &= \dfrac{c\cdot 4}{52} & c &\in \{ 0,\,1,\,\dots,\,13\},
	\end{align*}
	where $a,\,b,\,c$ are exactly the same as in the $n$-deck case since these coefficients are strictly based on the initial hand layouts. Thus, they also share the same relationship with respect to their sum. 
	
	Then, we have
	\begin{align*}
		\lim_{n \to \infty} P''(W_{s,n}) &= \lim_{n \to \infty} \dfrac{a\cdot 4n - w}{52n - 4} = \dfrac{a\cdot 4}{52} = P''(W_{s,\infty}),\\
		\lim_{n \to \infty} P''(T_{s,n}) &= \lim_{n \to \infty} \dfrac{b\cdot 4n - t}{52n - 4} = \dfrac{b\cdot 4}{52} = P''(T_{s,\infty}),\\
		\lim_{n \to \infty} P''(L_{s,n}) &= \lim_{n \to \infty} \dfrac{c\cdot 4n - \ell}{52n - 4} = \dfrac{c\cdot 4}{52} = P''(L_{s,\infty}),
	\end{align*}
	which establishes half of the statements.
	
	Now, assume that the player will hit. Then, when we use $n$ decks, 
	\begin{align*}
		P''(W_{h,n}) = \sum_{r = 1}^{11} P''(W_{h,n} \cap R_r) = \sum_{r=1}^{11} P''(W_{h,n} \vert R_r) \cdot P''(R_r) = \sum_{r=1}^{11} \dfrac{a_r\cdot 4n - w_r}{52n-5}\cdot\dfrac{4n-\gamma_r}{52n-4},
	\end{align*}
	where $R_r$ is the event that the value of player's hit-card is $r$. Here $a_r \in \{0,\,1,\,\dots,\,13\}$ is the number of ranks that will result in the player winning given that their hit-card value is $r$, $w_r \in \{0,1,\,2,\,3,\,4,\,5\}$ is the number of cards no longer in the deck what would result in the player winning given that their hit card value is $r$, and $\gamma_r \in \{0,\,1,\,2,\,3,\,4\}$ is the number of cards whose value is $r$ that are no longer in the deck. We emphasize that $a_r$ and $w_r$ are calculated assuming that the player has already received their hit-card. Furthermore,
	\[ P''(W_{h,n}) = \sum_{r=1}^{11} \dfrac{a_r\cdot 4n - w_r}{52n-5}\cdot\dfrac{4n-\gamma_r}{52n-4} =  \dfrac{A(4n)^2}{(52n-5)(52n-4)} + \mathcal{O}(1/n),\]
	where $A = \sum_r a_r$ and $\mathcal{O}(1/n)$ indicates a function that will converge to zero at the same rate as that of $\frac{1}{n}$ as $n \to \infty$. 
	
	Similarly,
	\begin{align*}
		P''(T_{h,n}) = \sum_{r=1}^{11} \dfrac{b_r\cdot 4n - t_r}{52n-5}\cdot\dfrac{4n-\gamma_r}{52n-4} =\dfrac{B(4n)^2}{(52n-5)(52n-4)} + \mathcal{O}(1/n),\\
		P''(L_{h,n}) = \sum_{r=1}^{11} \dfrac{c_r\cdot 4n - \ell_r}{52n-5}\cdot\dfrac{4n-\gamma_r}{52n-4} =\dfrac{C(4n)^2}{(52n-5)(52n-4)} + \mathcal{O}(1/n),
	\end{align*}
	where $B = \sum_r b_r$ and $C = \sum_r c_r$.  Additionally, we know $a_r + b_r + c_r = 13$ and $w_r + t_r + \ell_r = 5$ for each $r$, and $\sum_r \gamma_r = 4$. We also note that when the hit-card is an ace, we have one of $R_1$ or $R_{11}$ empty.  We also have that $\sum_r (a_r+b_r+c_r) = A+B+C = 169$.
	
	On the other hand, if we play using 1 deck with replacement, we have
	\begin{align*}
		P''(W_{h,\infty}) = \sum_{r = 1}^{11} \dfrac{a_r\cdot 4}{52} \cdot \dfrac{4}{52}  = A\bigg(\dfrac{4}{52}\bigg)^2,\\
		P''(T_{h,\infty}) = \sum_{r = 1}^{11} \dfrac{b_r\cdot 4}{52} \cdot \dfrac{4}{52} = B\bigg(\dfrac{4}{52}\bigg)^2,\\
		P''(L_{h,\infty}) = \sum_{r = 1}^{11} \dfrac{c_r\cdot 4}{52} \cdot \dfrac{4}{52} = C\bigg(\dfrac{4}{52}\bigg)^2.
	\end{align*}
	
	Finally, we have
	\begin{align*}
		\lim_{n \to \infty} P''(W_{h,n}) = \lim_{n \to \infty} \dfrac{A(4n)^2}{(52n-5)(52n-4)} + \mathcal{O}(1/n) = A\bigg(\dfrac{4}{52}\bigg)^2 = P''(W_{h,\infty}),\\
		\lim_{n \to \infty} P''(T_{h,n}) = \lim_{n \to \infty} \dfrac{B(4n)^2}{(52n-5)(52n-4)} + \mathcal{O}(1/n) = B\bigg(\dfrac{4}{52}\bigg)^2 = P''(T_{h,\infty}),\\
		\lim_{n \to \infty} P''(L_{h,n}) = \lim_{n \to \infty} \dfrac{C(4n)^2}{(52n-5)(52n-4)} + \mathcal{O}(1/n) = C\bigg(\dfrac{4}{52}\bigg)^2 = P''(L_{h,\infty}),
	\end{align*}
	which establishes the probabilities.
	
	Lastly,
	\begin{align*}
		E''[X_{s,n}] &= P''(W_{s,n}) - P''(L_{s,n}) \to P''(W_{s,\infty}) - P''(L_{s,\infty}) = E''[X_{s,\infty}],\\
		E''[X_{h,n}] &= P''(W_{h,n}) - P''(L_{h,n}) \to P''(W_{h,\infty}) - P''(L_{h,\infty}) = E''[X_{h,\infty}],
	\end{align*}
	establishes the result.	Note that if the player is dealt a natural, then the limits of the expectations still hold, but we would instead have $1.5\cdot P''(W_{\cdot,n}) - P''(L_{\cdot,n})\to 1.5\cdot P''(W_{\cdot,\infty}) - P''(L_{\cdot,\infty}) $.
\end{proof}

\begin{remark}
	Note that
	\begin{align*}
		P''(W_{s,\infty}) + P''(T_{s,\infty}) + P''(L_{s,\infty}) &= (a + b + c)\cdot\bigg(\dfrac{4}{52}\bigg) = 1,\\
		P''(W_{h,\infty}) + P''(T_{h,\infty}) + P''(L_{h,\infty}) &= (A + B + C)\cdot\bigg(\dfrac{4}{52}\bigg)^2 = 1.
	\end{align*}
	Similarly, 
	\begin{align*}
		P''(W_{s,n}) + P''(T_{s,n}) + P''(L_{s,n}) &= \dfrac{(a+b+c)\cdot4n - (w+t+l)}{52n - 4} =  1,\\
		P''(W_{h,n}) + P''(T_{h,n}) + P''(L_{h,n}) &= \sum_{r=1}^{11} \dfrac{\big((a_r+b_r+c_r)\cdot 4n - (w_r+t_r+l_r)\big)\cdot(4n - \gamma_r)}{(52n-5)(52n - 4)} = 1.
	\end{align*}
\end{remark}

Before we move on to proving the next result, we discuss the variety of distinct hand layouts that are possible when dealt two cards. Every hand layout fits into one of four categories. For each of these four categories of hand layout, we know both the number of hand layouts within each category, as well as how many permutations have that fixed hand layout.  We summarize this information below: for $n \ge 2$, we have
\begin{enumerate}
	\item Distinct card values, neither of which is valued at 10 (e.g., $(3,4)$, $(2,7)$, etc.): There are $36$ distinct hand layouts, and, within each, there are $2 \cdot 4n \cdot 4n = 32n^2$ permutations. 
	\item Non-distinct card values, neither of which is valued at 10 (e.g., $(3,3)$, $(4,4)$, etc.): There are $9$ distinct hand layouts, and, within each, there are $P(4n,2) = 16n^2 - 4n$ permutations.
	\item Distinct card values, exactly one of which is valued at 10 (e.g., $(3,10)$, $(4,10)$, etc.): There are $9$ distinct hand layouts, and, within each, there are $2\cdot 16n \cdot 4n = 128n^2$ permutations.
	\item Both cards are valued at 10 (e.g., $(J,K)$, $(10,Q)$, etc.): There is only one distinct hand layout with $P(16n,2) = 256n^2 - 16n$ permutations.
\end{enumerate}
That is, altogether, there are 55 distinct hand layouts with $P(52n,2)$ total permutations. For each fixed hand value, there exist constants $a,\,b,\,c,\,d$ such that the total number of permutations for that hand value is given by
\[ a\cdot 32n^2   + b\cdot (16n^2 - 4n) + c\cdot 128n^2 + d\cdot (256n^2 - 16n).\]
For example, for the hands valued at $14$, there are 
\[ 3\cdot 32n^2 + 1\cdot (16n^2-4n) + 1 \cdot 128n^2 + 0 \cdot (256n^2-16n). \]
That is, there are three layouts in category (1) - $\{ (A,3), \, (5,9),\, (6,8)\}$ - one layout in category (2) - $\{(7,7)\}$ - one layout in category (3) - $\{(4,10)\}$ - and no layouts in category (4). 

On the other hand, when we play using 1 deck with replacement, we have the following:
\begin{enumerate}
	\item Distinct card values, neither of which is valued at 10: There are $36$ distinct hand layouts, and, within each, there are $2\cdot 4 \cdot 4 = 32$ permutations. 
	\item Non-distinct card values, neither of which is valued at 10: There are $9$ distinct hand layouts, and, within each, there are $4\cdot 4 = 16$ permutations.
	\item Distinct card values, exactly one of which is valued at 10: There are $9$ distinct hand layouts, and, within each, there are $2\cdot 16 \cdot 4 = 128$ permutations.
	\item Both cards are valued at 10: There is only one distinct hand layout with $16\cdot 16 = 256$ permutations.
\end{enumerate}
That is, there remain the same 55 distinct hand layouts, but with 2704 total permutations.

Our goal now is to compute the probabilities of obtaining all permutations of hand layouts given a permutation of hand values for both $n$-deck and 1-deck with replacement variations. The general result follows:

\begin{proposition}
	\label{proposition:handvalue}
	Suppose that the initial hand values of the player and dealer are fixed while playing the two up-card variation and let $n$ be the number of decks. Let $i = 1,\,2,\,\dots,\, L$ and $j = 1,\,2,\,\dots,\,M$ be the number of distinct hand layouts for the player and dealer given their fixed starting totals, respectively.  Then for each $i$ and $j$, as $n \to \infty$, 
	\[ P'(E_{i,n} \cap F_{j,n}) \to P'(E_{i,\infty} \cap F_{j,\infty}), \]
	where $E_{i,\cdot}$ and $F_{j,\cdot}$ are events corresponding to distinct hand layouts for the player and dealer, respectively. 
\end{proposition}

Note that for each $i$, $E_{i,n}$ corresponds to a distinct hand layout for the player and is a collection of every two card permutation that is of this layout. In what follows, we use the absolute value notation $|E_{i,n}|$ to denote the number of permutations that are of the same hand layout. All distinct hand layouts $E_{i,n}$ have the same hand value for all $i$. Also, the layouts $E_{i,n}$ and $E_{i,\infty}$ are the same for all $n$, and the same is true of $F_{j,n}$ and $F_{j,\infty}$. 
\begin{proof}
	Let the initial hand values of the player and dealer be fixed. Then for any $i \in \{1,\,2,\,\dots,\, L\}$ and $j \in \{1,\,2,\,\dots,\,M\}$, we have
	\[ P'(E_{i,n} \cap F_{j,n}) = P'(F_{j,n} | E_{i,n})\cdot P'(E_{i,n}) .\]
	
	First, 
	\begin{align} 
		\label{eq:peijn}
		P'(E_{i,n}) = \dfrac{|E_{i,n}|}{a\cdot 32n^2   + b\cdot (16n^2-4n) + c\cdot 128n^2 + d\cdot (256n^2 - 16n) }. 
	\end{align}
	Note that $|E_{i,n}| \in \{32n^2, \, 16n^2 - 4n,\,128n^2,\, 256n^2-16n\}$, where each function grows at the rate of $n^2$ as $n \to \infty$ and counts the number of permutations of the player's hand layout, $E_{i,n}$. Furthermore, the constants $a,\,b,\,c,\,d$ are determined by the player's starting hand value and the relationship $\sum_i |E_{i,n}| = a\cdot 32n^2   + b\cdot (16n^2 - 4n) + c\cdot 128n^2 + d\cdot (256n^2 - 16n)$ follows by this definition.
	
	Next, 
	\begin{align}
		\label{eq:pfjn}
		P'(F_{j,n} | E_{i,n}) = \dfrac{|F_{j,n}| - h_{i,j,n}}{w\cdot 32n^2  + x \cdot (16n^2 - 4n) + y\cdot 128n^2 + z\cdot (256n^2 - 16n) - q_{i,n} } ,
		\end{align}
	where $|F_{j,n}| - h_{i,j,n}$ is the number of permutations of the dealer's hand layout, $F_{j,n}$, minus the permutations of this layout that consist of cards from the player's hand, $h_{i,j,n}$. The constants $w,\,x,\,y,\,z,$ are determined only by the dealer's hand value. Here, $q_{i,n}$ is the number of hands with the dealer's starting value that have been removed from the deck due to the cards that make up player's hand. The quantities $q_{i,n}$ and $h_{i,j,n}$ are both functions grow at the rate of $n$ as $n \to \infty$ such that $\sum_{j} h_{i,j,n} = q_{i,n}$. This can be shown by calculating the number of possible permutations in each category of hand layout when $1$ and $2$ cards have been removed from consideration. Also, we have $\sum_{j} |F_{j,n}| =w\cdot 32n^2  + x \cdot(16n^2-4n) + y\cdot 128n^2 + z\cdot (256n^2-16n)$ by definition. 
	
	On the other hand, when we play using 1 deck with replacement, we have
	\[ P'(E_{i,\infty} \cap F_{j,\infty}) = P'(E_{i,\infty})\cdot P'(F_{j,\infty}). \]
	Then, we have
	\begin{align*}
		P'(E_{i,\infty}) &= \dfrac{E_{i}}{a \cdot 32 + b \cdot 16 + c \cdot 128 + d \cdot 256}, \\
		P'(F_{j,\infty}) &= \dfrac{F_{j}}{w \cdot 32 + x \cdot 16 + y \cdot 128 + z \cdot 256},
		\end{align*}
	where $a,\,b,\,c,\,d$ and $w,\,x,\,y,\,z,$ are defined as in equations (\ref{eq:peijn}) and (\ref{eq:pfjn}), respectively, and $E_i$ and $F_j$ are the number of permutations for the hand layouts corresponding to $E_{i,\infty}$ and $F_{j,\infty}$, respectively. In particular, $E_i, F_j \in \{32, \,16,\, 128, \,256\}$. Then, by definition we have $\sum_i E_i = 32a + 16b + 128c + 256d$ and $\sum_j F_j = 32w + 16x + 128y + 256z$. 
	
	Lastly, when we send $n \to \infty$ in equations (\ref{eq:peijn}) and (\ref{eq:pfjn}), we have
	\begin{align*}
		\lim_{n \to \infty} P'(E_{i,n}) = \dfrac{E_i}{32a + 16b + 128c + 256d} = P'(E_{i,\infty}),\\
		\lim_{n \to \infty} P'(F_{j,n} | E_{i,n}) = \dfrac{F_j}{32w + 16x + 128y + 256z} = P'(F_{j,\infty}).
		\end{align*}
	Here we have used the fact that $\lim_{n\to\infty}|E_{i,n}|\cdot n^{-2} =  E_i$, which follows from the definitions of these quantities.  In particular, $\lim_{n \to \infty} |E_{i,n}|\cdot n^{-2}$ is exactly one of the constants in $\{32,\, 16,\, 128,\, 256 \}$. Similarly, we have $\lim_{n \to \infty}(|F_{j,n}|-h_{i,j,n})\cdot n^{-2} = F_j$ and $\lim_{n\to\infty} q_{i,n} \cdot n^{-2} = 0$. 
	
	Thus, we have
	\[ \lim_{n \to \infty} P'(E_{i,n} \cap F_{j,n}) = \lim_{n\to \infty} P'(F_{j,n} | E_{i,n})\cdot P'(E_{i,n}) = P'(E_{i,\infty})\cdot P'(F_{j,\infty}) = P'(E_{i,\infty} \cap F_{j,\infty})\]
	for each $i = 1,\,2,\, \dots,\, L$ and $j = 1,\,2,\,\dots,\,M$.
\end{proof}

As a corollary to Propositions \ref{proposition:handlayout} and \ref{proposition:handvalue} we are able to calculate the probability of winning, pushing, and losing and expectations conditioned only on the starting hand totals of the player and the dealer - instead of being conditioned on the starting hand layouts.

\begin{corollary}
	\label{corollary:level2}
	Suppose that the initial hand values of the player and dealer are fixed. Then, as $n \to \infty$, 
	\begin{align*}
		P'(W_{h,n}) & = \sum_{i,j} P'(E_{i,n} \cap F_{j,n}) \cdot P''(W_{h,n}) \to \sum_{i,j} P'(E_{i,\infty} \cap F_{j,\infty}) \cdot P''(W_{h,\infty}) = P'(W_{h,\infty}),\\
		P'(T_{h,n}) & = \sum_{i,j} P'(E_{i,n} \cap F_{j,n}) \cdot P''(T_{h,n}) \to \sum_{i,j} P'(E_{i,\infty} \cap F_{j,\infty}) \cdot P''(T_{h,\infty}) = P'(T_{h,\infty}),\\
		P'(L_{h,n}) & = \sum_{i,j} P'(E_{i,n} \cap F_{j,n}) \cdot P''(L_{h,n}) \to \sum_{i,j} P'(E_{i,\infty} \cap F_{j,\infty}) \cdot P''(L_{h,\infty}) = P'(L_{h,\infty}),\\
		E'[X_{h,n}] &= \sum_{i,j} P'(E_{i,n} \cap F_{j,n}) \cdot E''[X_{h,n}] \to \sum_{i,j} P'(E_{i,\infty}\cap F_{j,\infty}) \cdot E''[X_{h,\infty}] = E'[X_{h,\infty}],\\
		E'[X_{s,n}] &= \sum_{i,j} P'(E_{i,n} \cap F_{j,n}) \cdot E''[X_{s,n}] \to \sum_{i,j} P'(E_{i,\infty}\cap F_{j,\infty}) \cdot E''[X_{s,\infty}] = E'[X_{s,\infty}].
	\end{align*}
	Furthermore, there is a positive integer $N$ sufficiently large such that for any $n \ge N$, the basic strategy for $n$ decks is the same as the basic strategy using 1 deck with replacement. 
\end{corollary}

\begin{remark}\label{remark:basicstrat}
	It is based on the result of Corollary \ref{corollary:level2}, involving the expected values, that we are able to generate the basic strategy for $n$ deck and 1 deck with replacement variations.  That is, given starting hand totals for the player and dealer, we compare $E'[X_{h,\cdot}]$ with $E'[X_{s,\cdot}]$ and choose hit or stand based on whichever is larger. Given the limiting of the $E'$ quantities, we must have that the basic strategy eventually stays fixed for a large enough number of decks.
	\end{remark}
We emphasize again that all probabilities involving $P''$ are implicitly conditioned on the specific hand layouts $E_{i, \cdot}$ and $F_{j, \cdot}$ despite not being indexed by $i$ or $j$. We remove the reference to $i$ and $j$ at the level of $P''$ to maintain a lower level of complexity in terms of the indices used.

Our final goal is to generate the overall probabilities of winning, pushing, and losing, as well as the expectation, given that we play according to the basic strategy determined by Corollary \ref{corollary:level2}. We begin by calculating the probabilities associated to the player and dealer receiving all permutations of initial starting totals. These probabilities are not conditioned on any prior information. 

\begin{proposition}
	\label{proposition:nocondition}
	Suppose we are playing the two up-card variation of the game and let $n$ be the number of decks. Also, let $x$ and $y$ be the starting hand totals of the player and dealer, respectively. Then, as $n \to \infty$, 
	\begin{align*}
		P(A_{x,n} \cap B_{y,n}) &\to P(A_{x,\infty} \cap B_{y,\infty}), 
	\end{align*}
	where $A_{x,\cdot}$ and $B_{y,\cdot}$ consist of all starting hands for the player and dealer whose values are $x$ and $y$, respectively.  
\end{proposition}

\begin{proof}
	Let $x$ and $y$ be the starting totals of the player and dealer, respectively. Then,
	\begin{align*}
		P(A_{x,n} \cap B_{y,n}) &= P(B_{y,n} | A_{x,n}) \cdot P(A_{x,n}) \\
		& = \dfrac{32n^2w  +  (16n^2 - 4n)x + 128n^2y +(256n^2 - 16n) z - q_{i,n} }{(52n)^2-260n+6}\\
		&\hspace{20pt} \cdot \dfrac{32n^2a +  (16n^2 - 4n)b + 128n^2c + (256n^2 - 16n)d }{(52n)^2-52n} ,
		\end{align*}
	where $a,\,b,\,c,\,d$ and $w,\,x,\,y,\,z$ and $q_{i,n} = \mathcal{O}(n)$ are defined as in Proposition \ref{proposition:handvalue}. In limit, we have
	\begin{align*}
		\lim_{n \to \infty} P(A_{x,n} \cap B_{y,n}) &= \dfrac{32w + 16x + 128y + 16z}{2704}\cdot\dfrac{32a+16b + 128c+256d}{2704} \\
		&= P(A_{x,\infty})\cdot P(B_{y,\infty}) \\
		&= P(A_{x,\infty} \cap B_{y,\infty}),
	\end{align*}
	which establishes the result.
\end{proof}

Finally, we are able to establish the main result of the section as a corollary to the above propositions. In particular, the overall probabilities of winning, pushing, and losing and expectation when using $n$ decks limit to the corresponding quantities using 1 deck with replacement as the number of decks approaches infinity. 

\begin{theorem}
	\label{theorem:main}
	Let $X_n$ and $X_\infty$ be the payouts when the player plays using the basic strategy dictated by Remark \ref{remark:basicstrat}. Also, let $W_n,\, T_n,\, L_n$ and $W_\infty,\, T_\infty, \, L_\infty$ be the events of the player winning, pushing, and losing when using basic strategy with $n$ decks and 1 deck with replacement, respectively. Then, as $n \to \infty$, 
	\begin{align*}
		P(W_n) & = \sum_{x,y} P(A_{x,n} \cap B_{y,n}) \cdot P'(W_{n}) \to \sum_{x,y} P(A_{x,\infty} \cap B_{y,\infty})\cdot P'(W_{\infty}) = P(W_\infty),\\
		P(T_n) & = \sum_{x,y} P(A_{x,n} \cap B_{y,n}) \cdot P'(T_{n}) \to \sum_{x,y} P(A_{x,\infty} \cap B_{y,\infty})\cdot P'(T_{\infty}) = P(T_\infty),\\
		P(L_n) & = \sum_{x,y} P(A_{x,n} \cap B_{y,n}) \cdot P'(L_{n}) \to \sum_{x,y} P(A_{x,\infty} \cap B_{y,\infty})\cdot P'(L_{\infty}) = P(L_\infty),\\
		E[X_{n}] &= \sum_{x,y} P(A_{x,n} \cap B_{y,n}) \cdot E'[X_{n}] \to \sum_{x,y} P(A_{x,\infty} \cap B_{y,\infty}) \cdot E'[X_\infty] = E[X_{\infty}]. 
		\end{align*}	
\end{theorem}
Note that the quantities $P'(W_{n}),\, P'(T_{n}),\,P'(L_{n})$ and $P'(W_{\infty}),\, P'(T_{\infty}),\,P'(L_{\infty})$ are all implicitly conditioned on the starting totals of the player and dealer, $A_{x,n} \cap B_{y,n}$, despite not being indexed by $x$ and $y$. 

\begin{remark}
	In the above argument, we assumed that we were playing the two up-card variation of the game, but our results are easily transferable to the other two dealer information variations. One can consider these overall results as a procedure of three stages: complete knowledge, partial knowledge and no knowledge of the player and dealer hands (as described in Appendix \ref{sec:derivation}). However, when we are playing the one up-card variation, there will be an added step in the partial knowledge stage, where we only know the value of the dealer's up-card and the value of the player's hand. It is again at this stage where the optimal decision is determined. Similarly, when we are playing the no up-card variation, there is a further added step at the partial knowledge stage, where we only know the value of the player's hand. These extra calculations at this stage are done in a very similar way to the above.
\end{remark}

\subsection{Overall Probabilities and Expectation using 1 Deck With Replacement}
\label{sec:1dwr}

Based on the calculations in Section \ref{sec:longtermbehavior}, we are able to generate the basic strategy for each of our variations. When we employ this strategy, we can summarize the overall probabilities and expectations, and this is done in Table \ref{table:1dwr} below.
\begin{table}[h!]	
	\setlength{\arrayrulewidth}{0.5mm}
	\centering{
	\begin{tabular}{|c|c|c|c|}	
		\hline
		\rowcolor{RoyalBlue!20} 
		\textbf{Variation:} & \textbf{Two up-card} & \textbf{One up-card} & \textbf{No up-card} \\ \hline
		\cellcolor{RoyalBlue!20}$P(W_\infty)$      & 0.46935  & 0.45413  & 0.45282      \\ \hline
		\cellcolor{RoyalBlue!20}$P(L_\infty)$      & 0.44543  & 0.46392  & 0.46698      \\ \hline
		\cellcolor{RoyalBlue!20}$P(T_\infty)$      & 0.08522  & 0.08195  & 0.08020     \\ \hline
		\cellcolor{RoyalBlue!20}$E[X_\infty]$      & 0.046467 & 0.012959 & 0.008384   \\ \hline
	\end{tabular} }
\caption{The entries of the table give the overall probabilities of winning, losing, and pushing and expectation when we play using one deck with replacement for each of the three deck composition variations.}
\label{table:1dwr}
\end{table}

When we compare the overall results regarding our three deck composition variations in Theorem \ref{theorem:deckcomp} with the results in Table \ref{table:1dwr} regarding the one deck without replacement results, we can see that there is not a significant change. However, it is our understanding that casinos use more decks primarily as a preventative measure for card counting, as well as to reduce the frequency of shuffling to keep the rate of play high. We also suspect that in traditional Blackjack, the overall expectation goes down too when the number of decks goes up. Interestingly, when we play using one deck with replacement, card counting is actually not possible since the distribution of all card values is fixed no matter how many cards are in play. That is, if one could \textit{practically} play Blackjack using one deck with replacement, card counting would be eliminated, while simultaneously providing a larger edge to the casino.

%%%%%%%%%%%%%%%%%%%%%%%%%%%%%%%%%%%%% APPENDIX

\appendix

\section{Derivation of Results} % Section 3
\label{sec:derivation}

In this section of the paper, we outline how our program generates the results provided in Sections \ref{sec:computations} and \ref{sec:deckcomp}, as well as provide a concrete example to elucidate the abstract computations done within the body of the paper. We will again limit our discussion to the one deck, two up-card variation of the game for concreteness, but the core principles are unchanged for the other variations. The procedure for generating these results comes in three stages:
\begin{enumerate}[1.]
	\item \textbf{Complete knowledge}: Specific hand layouts for the player and dealer are known. 
	\item \textbf{Partial knowledge}: Only the hand values of the player and dealer are known.
	\item \textbf{No knowledge}: No information is known about the hand of the player or dealer.
\end{enumerate}
As we progress through the three stages, the program utilizes loops to fill in missing information.  For example, at stage 3, loops are used to cycle through all permutations of starting hand totals for the player and dealer. Within this loop, at stage 2, there are loops that cycle through all permutations of distinct hand layouts for the player and dealer given their respective starting hand values. 

At stage 1, the cards for both player and dealer have been dealt and are considered known information. When the player or dealer hit, a sequence of loops are used to cycle through the remaining cards and temporarily add them to the hand. In doing so, we are able to count the number of cards that will result in the player winning, losing, and pushing. Once these are used to calculate the probabilities for each of the two decisions, the expected values are computed. These probabilities and expectations are carried through all three stages of the procedure where they will be weighed by how likely this particular deal of cards was to occur. This strategy is what we used in the proof of Proposition \ref{proposition:handlayout}.

The basic strategy is calculated at stage 2. We first compute the probability of getting a specific hand layout given a particular starting total for both player and dealer. This is done by counting the number of combinations of each distinct hand layout and dividing by the total number of combinations that have this starting total. We emphasize that at this level, we do not know the specific cards that the player and dealer have, but rather only their starting hand values. Then, we use these quantities to weigh the stage 1 probabilities and expectations for each decision. That is, given the starting totals for the player and dealer, the larger of the two weighed expectations gives us the optimal decision. This is the approach taken in the proof of Proposition \ref{proposition:handvalue} and the resulting Corollary \ref{corollary:level2}.

Lastly, in stage 3, we obtain the overall probabilities and expectations when we play according to the basic strategy determined at the end of stage 2. The starting point for this stage is to determine the probability for every possible permutation of starting hand totals for the player and dealer. Once we have these, we use them to weigh the \textit{optimal} stage 2 probabilities and expectations.  The sum of these weighed quantities give us the overall results. This process is done in the proof of Proposition \ref{proposition:nocondition} and the main theoretical result of the paper Theorem \ref{theorem:main}.

For clarity, we demonstrate the computations at each stage concretely in the following Example \ref{ex:stage1} below.

\begin{example}\label{ex:stage1}
Suppose the player and dealer are both dealt two cards face-up from a single deck of cards. That is, we are in stage 1, described above. In particular, assume:
\begin{itemize}
	\item  The player is dealt: $(A,2)$.  That is, the player has a soft hand value of 13. 
	\item  The dealer is dealt: $(6,8)$.  That is, the dealer has a hard hand value of 14. 
\end{itemize}
Should the player hit, or should they stand?

To determine the optimal decision, we must first compute the probability that the player will win, lose, and push when they hit, and when they stand. Note here that the dealer \textit{must} hit, since their total is less than 17.

Suppose first that the player will \textit{stand} with 13. There are only two possible outcomes - win and lose - since a push is not possible in this scenario.  The player will win only when the dealer busts.  That is, the dealer's hit-card is valued at 8 or higher. Every dealer hit-card valued 7 or lower will result in the player losing. Since there are 48 cards remaining in the deck, we have the following probabilities:
\[ P''(W_s) = \dfrac{6\cdot 4 - 1}{52 - 4} = \dfrac{23}{48}, \hspace{20pt} P''(L_s) = \dfrac{7\cdot 4 - 3}{52 - 4} = \dfrac{25}{48}, \hspace{20pt} P''(T_s) = \dfrac{0\cdot 4 - 0}{52 - 4} = \dfrac{0}{48} .\]
Here $W_s$, $L_s$, and $T_s$ are the events that the player wins, loses, and pushes. We use the $''$ to signify that these probabilities are dependent on the initial hand layouts that are given. These probabilities are calculated by multiplying the number of ranks whose value will result in the player winning, losing, and tying by four, and then subtracting out the number of cards within these ranks that are no longer in the deck. We generalize this in Section \ref{sec:deckcomp}, where we adopt an analogous notation. Therefore, the expected value of a $\$1$ bet for the player when they stand is
\[ E''[X_{s}] = 1\cdot \dfrac{23}{48} + (-1) \cdot \dfrac{25}{48} = -\dfrac{2}{48} \approx -0.0417 .\]

Suppose now that the player will \textit{hit}. To calculate the probability of winning, losing, and pushing, we must first condition on what the player's hit-card is. Once we know the player's hit card, the computation for the probabilities are the same as in the previous paragraph. For example, suppose that the player's hit-card is an ace. Then, the player's new total is 14, and 
\begin{align*}
	P''(W_h \cap R_1) &= P''(W_h | R_1)\cdot P''(R_1) = \dfrac{6 \cdot 4 - 1}{52 - 5}\cdot \dfrac{4 - 1}{52 - 4} =  \dfrac{69}{2256},  \\
	P''(L_h \cap R_1) &= P''(L_h | R_1)\cdot P''(R_1) = \dfrac{7\cdot 4 - 4}{52 - 5}\cdot \dfrac{4-1}{52-4} =  \dfrac{72}{2256} , \\
	P''(T_h \cap R_1) &= P''(T_h | R_1)\cdot P''(R_1) = \dfrac{0\cdot 4 - 0}{52 - 5}\cdot \dfrac{4 - 1}{52 - 4} =  \dfrac{0}{2256},  
\end{align*}
where $W_h, \, L_h, \, T_h$ are the events that the player wins, loses, and pushes when they hit, respectively. Also, $R_r$ represents the event that the player's hit-card is valued at $r$. We repeat this for every possible value that the player's hit-card can take. We summarize the results in Table \ref{table:numerators} below.  

\begin{table}[h!]
	\centering
	\setlength{\arrayrulewidth}{0.5mm}
	\begin{tabular}{|c|c|c|c|c|c|c|c|c|c|c|c|}
		\hline
		\rowcolor{RoyalBlue!20}
		\cellcolor{RoyalBlue!20} \textbf{Hit-card:} & A  & 2  & 3   & 4   & 5   & 6   & 7   & 8   & 9   & 10 & Total \\ \hline
		\cellcolor{RoyalBlue!20}$W$ & 69 & 69 & 104 & 116 & 132 & 111 & 164 & 129 & 88   & 352 & 1334 \\ \hline
		\cellcolor{RoyalBlue!20}$L$ & 72 & 63 & 72  & 56  & 40  & 18  & 12  & 0   & 100 & 400 & 833 \\ \hline
		\cellcolor{RoyalBlue!20}$T$ & 0  & 9  & 12  & 16  & 16  & 12  & 12  & 12  & 0   & 0  & 89 \\ \hline
	\end{tabular}
	\caption{The entries of the table give the numerator in the probability of the intersection between the row and column with same denominator of 2256.}
	\label{table:numerators}
\end{table}

What we have done is split up the probability of winning when the player hits into ten mutually exclusive events according to the possible values the player's hit-card can take. Then we used the conditional probability formula to calculate the probabilities of the individual mutually exclusive events.  Then,
\begin{align*}
	P''(W_h) &= P''(W_h \cap R_1) + P''(W_h \cap R_2) + \cdots + P''(W_h \cap R_{10})  = \dfrac{1334}{2256},\\
	P''(L_h) &= P''(L_h \cap R_1) + P''(L_h \cap R_2) + \cdots + P''(L_h \cap R_{10}) = \dfrac{833}{2256},\\
	P''(T_h) &= P''(T_h \cap R_1) + P''(T_h \cap R_2) + \cdots + P''(T_h \cap R_{10}) = \dfrac{89}{2256}.
\end{align*}
Thus, the expected value of a $\$1$ bet when we hit is given by
\[ E''[X_{h}] = 1\cdot \dfrac{1334}{2256} + (-1)\cdot\dfrac{833}{2256} = \dfrac{501}{2256} \approx 0.2221. \]

Moving to stage 2, we will be assuming that the player has been dealt a soft 13 and the dealer a hard 14. What we calculate at this level are the probabilities of being dealt each permutation of hand layouts that are possible given these starting hand totals. In the notation of Proposition \ref{proposition:handvalue}, this means $E_{i,1}$ is the collection of all hands valued at soft 13, of which there is only 1 distinct hand layout, $(A, 2)$ (i.e. $L = 1$).  Because there are $2\cdot4\cdot4$ combinations of this layout, we have $|E_{i,1}| = 32$ and thus $P'(E_{i,1}) = \frac{32}{32} = 1$. 

Next, there are 4 distinct hand layouts that will produce a hand total of hard 14: $(4,10)$, $(5,9)$, $(6,8)$, and $(7,7)$ (i.e. $M = 4$). In the notation of Proposition \ref{proposition:handvalue} we have:
\begin{itemize}
	\item  $F_{1,1}$: Consists of all $|F_{1,1}| = 128$ permutations of $(4,10)$.
	\item  $F_{2,1}$: Consists of all $|F_{2,1}| = 32$ permutations of $(5,9)$.
	\item  $F_{3,1}$: Consists of all $|F_{3,1}| = 32$ permutations of $(6,8)$.
	\item  $F_{4,1}$: Consists of all $|F_{4,1}| = 12$ permutations of $(7,7)$.
\end{itemize}
That is, there are $128 + 32 + 32 + 12 = 204$ total permutations of these four distinct hand layouts valued at hard 14. Note that in $(4,10)$ we mean any card whose value is 10 (10, Jack, Queen, King), not strictly rank 10, necessarily. Additionally, since there are no hard 14 hands that consist of an Ace or a 2, we have independence of $E_{i,1}$ and $F_{j,1}$ for each $i$ and $j$.  Thus, 
\[ P'(E_{i,1} \cap F_{j,1}) = P'(E_{i,1})\cdot P'(F_{j,1}) .\]
For example, 
\[P'(E_{1,1} \cap F_{3,1}) = \dfrac{32}{32} \cdot \dfrac{32}{204} = \dfrac{8}{51}.\] 
However, we emphasize that we do not always have independence, and therefore the above probability of this intersection will typically require a conditional probability, as is done in Proposition \ref{proposition:handvalue}.  The following table summarizes all probabilities of the intersections of $E_{i,1}$ and $F_{j,1}$:

\begin{table}[h!]
	\centering
	\setlength{\arrayrulewidth}{0.5mm}
	\renewcommand{\arraystretch}{2}
	\begin{tabular}{|c|c|c|c|c|}
	\hline
	\cellcolor{RoyalBlue!20}\backslashbox{\textbf{\textbf{Player}}}{\textbf{Dealer}}	& \cellcolor{RoyalBlue!20}$F_{1,1}$ &\cellcolor{RoyalBlue!20} $F_{2,1}$ & \cellcolor{RoyalBlue!20}$F_{3,1}$ & \cellcolor{RoyalBlue!20}$F_{4,1}$ \\ 
	\hline
	\cellcolor{RoyalBlue!20}$E_{1,1}$ & $\dfrac{32}{51}$ & $\dfrac{8}{51}$ & $\dfrac{8}{51}$ & $\dfrac{3}{51}$  \\ 
	\hline
	\end{tabular}
	\caption{The entries of the table give the probability $P(E_{i,1} \cap  F_{j,1})$, where $E_{i,1}$ and $F_{j,1}$ are the events corresponding to the distinct hand layouts for the player and dealer, respectively, given the initial starting hand totals.}
	\label{table:stage2}
\end{table}

In the following table, we have summarized all stage 1 expectations for hitting and standing given the starting totals for the player and dealer:

\begin{table}[h!]
	\centering
	\setlength{\arrayrulewidth}{0.5mm}
	\renewcommand{\arraystretch}{2}
	\begin{tabular}{|c|c|c|c|c|}
		\hline
		\cellcolor{RoyalBlue!20} \textbf{Hand Layouts:}	& \cellcolor{RoyalBlue!20} \(E_{1,1} \cap F_{1,1}\) &\cellcolor{RoyalBlue!20} $E_{1,1} \cap F_{2,1}$  & \cellcolor{RoyalBlue!20} $E_{1,1} \cap F_{3,1}$ & \cellcolor{RoyalBlue!20}$E_{1,1} \cap F_{4,1}$ \\ 
		\hline
		\cellcolor{RoyalBlue!20}$E''[X_s]$ & -$\dfrac{2}{48}$ & -$\dfrac{2}{48}$ & -$\dfrac{2}{48}$ & $0$  \\ 
		\hline
		\cellcolor{RoyalBlue!20}$E''[X_h]$ & $\dfrac{547}{2256}$ & $\dfrac{547}{2256}$ & $\dfrac{501}{2256}$ & $\dfrac{597}{2256}$  \\ 
		\hline
	\end{tabular}
	\caption{The entries of the table give the expectations when the player hits and stands given the particular permutation of hand layouts.}
	\label{table:stage1full}
\end{table}

To determine the optimal strategy when the player has a soft 13 and the dealer has a hard 14, we must weigh each stage 1 expected value from Table \ref{table:stage1full} by the corresponding weight from Table \ref{table:stage2} and add together.  In particular, we have
\begin{align*}
	E'[X_s] &= \dfrac{32}{51}\cdot \Big(-\dfrac{2}{48}\Big) + 2\cdot \dfrac{8}{51}\cdot\Big( - \dfrac{2}{48}\Big)  + \dfrac{3}{51}\cdot 0 = - \dfrac{128}{2448} \approx -0.05229,\\
	E'[X_h] &= \dfrac{32}{51}\cdot \dfrac{547}{2256} + \dfrac{8}{51}\cdot \Big(\dfrac{547}{2256} + \dfrac{501}{2256}\Big) +\dfrac{3}{51}\cdot \dfrac{597}{2256} =  \dfrac{27679}{115056} \approx 0.23369.
\end{align*} 

Lastly, given that the correct decision is to hit, we get the overall contribution of the expected value by multiplying this optimal expectation by the probability that we are dealt this permutation of hand values. Following the notation and strategy in the proof of Proposition \ref{proposition:nocondition}, we have
\[ P(A_{13,1} \cap B_{14,1}) = \dfrac{204}{2450}\cdot \dfrac{32}{2652} = \dfrac{816}{812175} ,\]
where we note that the index 13 in this context means \textit{soft} 13 and the 14 means \textit{hard} 14. Thus, when the player has a soft 13 and the dealer has a hard 14, the correct decision is to hit. Therefore, the contribution to the overall expected value is
\[ P(A_{13,1} \cap B_{14,1}) \cdot E'[X_h] \approx 0.0002417 .\]

When we repeat this process and sum all expected value contributions, we get the overall expectation given in Theorem \ref{theorem:deckcomp}.

\end{example}

\section{Basic Strategy Tables}\label{appendix:tables}
In this section, we present the basic strategy tables, which outline the optimal decision in each betting situation for our variants of the game. We remark, however, that the basic strategy tables that follow are not fully optimized. In particular, there are cases - denoted by an asterisk in the tables - in which a more optimal strategy is available, but it depends on knowing the specific hand layout the player or dealer starts with, rather than just the starting totals. For example in the two up-card variation, when the player's starting total is a hard 12 and the dealer's starting total is a hard 14, there are $5\cdot 4= 20$ distinct permutations of hand layouts. In 8 of these permutations the optimal strategy is to stand, and in the remaining 12 the optimal strategy is to hit. Despite this split in strategy, the expected value when the player always hits versus when they always stand favor standing. Notice also that even though we hit in more situations than we stand, the optimal decision is to stand.  This is because the expected value is weighed for each hand layout based on its probability of occurring.  That is, for example, the expected value is more heavily weighed by $(10,2)$ than it is by $(6,6)$, since the former is more likely to occur. Because our expected value depends on making one decision given a specific hand value - not hand layout - the expected value would increase if we were to operate with this level of detail. We have opted for brevity and consistency with how basic strategy is presented for traditional Blackjack.

In our basic strategy tables, we have also separated hard from soft starting hand totals. This is because the strategy for a soft hand is fundamentally different in most cases due to the fact that a soft hand cannot bust in our game - since only one action is allowed. This evidently incentivizes a more aggressive strategy.

\subsection{Two Up-Card}
The following two tables give the basic strategy for the two up-card variation played using one deck.  It is broken into two pieces due to size restrictions - first table considering only hard totals for the dealer and soft totals for the dealer in the second table.
\begin{table}[h]
	\setlength{\arrayrulewidth}{0.5mm}
	\begin{tabular}{|l|c|c|c|c|c|c|c|c|c|c|c|c|c|c|c|c|c|}
		\hline
		\rowcolor{RoyalBlue!20} \backslashbox{\textbf{\textbf{Player}}}{\textbf{Dealer}} & \textbf{4} & \textbf{5} & \textbf{6} & \textbf{7} & \textbf{8} & \textbf{9} & \textbf{10} & \textbf{11} & \textbf{12} & \textbf{13} & \textbf{14} & \textbf{15} & \textbf{16} & \textbf{17} & \textbf{18} & \textbf{19} & \textbf{20} \\ \hline
		\cellcolor{RoyalBlue!20} \textbf{Hard 4 - 11}      & \cellcolor{green!20}H & \cellcolor{green!20}H & \cellcolor{green!20}H & \cellcolor{green!20}H & \cellcolor{green!20}H & \cellcolor{green!20}H & \cellcolor{green!20}H  & \cellcolor{green!20}H  & \cellcolor{green!20}H  & \cellcolor{green!20}H  & \cellcolor{green!20}H  & \cellcolor{green!20}H  & \cellcolor{green!20}H  & \cellcolor{green!20}H  & \cellcolor{green!20}H  & \cellcolor{green!20}H  & \cellcolor{green!20}H  \\ \hline
		\cellcolor{RoyalBlue!20} \textbf{Hard 12}          & \cellcolor{green!20}H & \cellcolor{green!20}H & \cellcolor{green!20}H & \cellcolor{green!20}H & \cellcolor{green!20}H & \cellcolor{green!20}H & \cellcolor{green!20}H  & \cellcolor{green!20}H  & \cellcolor{green!20}H  & \cellcolor{green!20}H  & \cellcolor{green!20}H* & \cellcolor{red!20}S  & \cellcolor{red!20}S  & \cellcolor{green!20}H  & \cellcolor{green!20}H  & \cellcolor{green!20}H  & \cellcolor{green!20}H  \\ \hline
		\cellcolor{RoyalBlue!20} \textbf{Hard 13}          & \cellcolor{green!20}H & \cellcolor{green!20}H & \cellcolor{green!20}H & \cellcolor{green!20}H & \cellcolor{green!20}H & \cellcolor{green!20}H & \cellcolor{green!20}H  & \cellcolor{green!20}H  & \cellcolor{green!20}H  & \cellcolor{green!20}H & \cellcolor{red!20}S  & \cellcolor{red!20}S  & \cellcolor{red!20}S  & \cellcolor{green!20}H  & \cellcolor{green!20}H  & \cellcolor{green!20}H  & \cellcolor{green!20}H  \\ \hline
		\cellcolor{RoyalBlue!20} \textbf{Hard 14}          & \cellcolor{red!20}S & \cellcolor{red!20}S & \cellcolor{red!20}S & \cellcolor{green!20}H & \cellcolor{green!20}H & \cellcolor{green!20}H & \cellcolor{green!20}H  & \cellcolor{green!20}H  & \cellcolor{red!20}S  & \cellcolor{red!20}S  & \cellcolor{red!20}S  & \cellcolor{red!20}S  & \cellcolor{red!20}S  & \cellcolor{green!20}H  & \cellcolor{green!20}H  & \cellcolor{green!20}H  & \cellcolor{green!20}H  \\ \hline
		\cellcolor{RoyalBlue!20} \textbf{Hard 15}          & \cellcolor{red!20}S & \cellcolor{red!20}S & \cellcolor{red!20}S & \cellcolor{red!20}S & \cellcolor{red!20}S & \cellcolor{red!20}S & \cellcolor{red!20}S* & \cellcolor{red!20}S  & \cellcolor{red!20}S  & \cellcolor{red!20}S  & \cellcolor{red!20}S  & \cellcolor{red!20}S  & \cellcolor{red!20}S  & \cellcolor{green!20}H  & \cellcolor{green!20}H  & \cellcolor{green!20}H  & \cellcolor{green!20}H  \\ \hline
		\cellcolor{RoyalBlue!20} \textbf{Hard 16}          & \cellcolor{red!20}S & \cellcolor{red!20}S & \cellcolor{red!20}S & \cellcolor{red!20}S & \cellcolor{red!20}S & \cellcolor{red!20}S & \cellcolor{red!20}S  & \cellcolor{red!20}S  & \cellcolor{red!20}S  & \cellcolor{red!20}S  & \cellcolor{red!20}S  & \cellcolor{red!20}S  & \cellcolor{red!20}S  & \cellcolor{green!20}H  & \cellcolor{green!20}H  & \cellcolor{green!20}H  & \cellcolor{green!20}H  \\ \hline
		\cellcolor{RoyalBlue!20} \textbf{Hard 17}          & \cellcolor{red!20}S & \cellcolor{red!20}S & \cellcolor{red!20}S & \cellcolor{red!20}S & \cellcolor{red!20}S & \cellcolor{red!20}S & \cellcolor{red!20}S  & \cellcolor{red!20}S  & \cellcolor{red!20}S  & \cellcolor{red!20}S  & \cellcolor{red!20}S  & \cellcolor{red!20}S  & \cellcolor{red!20}S  & \cellcolor{red!20}S  & \cellcolor{green!20}H  & \cellcolor{green!20}H  & \cellcolor{green!20}H  \\ \hline
		\cellcolor{RoyalBlue!20} \textbf{Hard and Soft 18} & \cellcolor{red!20}S & \cellcolor{red!20}S & \cellcolor{red!20}S & \cellcolor{red!20}S & \cellcolor{red!20}S & \cellcolor{red!20}S & \cellcolor{red!20}S  & \cellcolor{red!20}S  & \cellcolor{red!20}S  & \cellcolor{red!20}S  & \cellcolor{red!20}S  & \cellcolor{red!20}S  & \cellcolor{red!20}S  & \cellcolor{red!20}S  & \cellcolor{red!20}S  & \cellcolor{green!20}H  & \cellcolor{green!20}H  \\ \hline
		\cellcolor{RoyalBlue!20} \textbf{Hard and Soft 19} & \cellcolor{red!20}S & \cellcolor{red!20}S & \cellcolor{red!20}S & \cellcolor{red!20}S & \cellcolor{red!20}S & \cellcolor{red!20}S & \cellcolor{red!20}S  & \cellcolor{red!20}S  & \cellcolor{red!20}S  & \cellcolor{red!20}S  & \cellcolor{red!20}S  & \cellcolor{red!20}S  & \cellcolor{red!20}S  & \cellcolor{red!20}S  & \cellcolor{red!20}S  & \cellcolor{red!20}S  & \cellcolor{green!20}H  \\ \hline
		\cellcolor{RoyalBlue!20} \textbf{Hard and Soft 20} & \cellcolor{red!20}S & \cellcolor{red!20}S & \cellcolor{red!20}S & \cellcolor{red!20}S & \cellcolor{red!20}S & \cellcolor{red!20}S & \cellcolor{red!20}S  & \cellcolor{red!20}S  & \cellcolor{red!20}S  & \cellcolor{red!20}S  & \cellcolor{red!20}S  & \cellcolor{red!20}S  & \cellcolor{red!20}S  & \cellcolor{red!20}S  & \cellcolor{red!20}S  & \cellcolor{red!20}S  & \cellcolor{red!20}S  \\ \hline
		\cellcolor{RoyalBlue!20} \textbf{Soft 12 - 14}     & \cellcolor{green!20}H & \cellcolor{green!20}H & \cellcolor{green!20}H & \cellcolor{green!20}H & \cellcolor{green!20}H & \cellcolor{green!20}H & \cellcolor{green!20}H  & \cellcolor{green!20}H  & \cellcolor{green!20}H  & \cellcolor{green!20}H  & \cellcolor{green!20}H  & \cellcolor{green!20}H  & \cellcolor{green!20}H  & \cellcolor{green!20}H  & \cellcolor{green!20}H  & \cellcolor{green!20}H  & \cellcolor{green!20}H  \\ \hline
		\cellcolor{RoyalBlue!20} \textbf{Soft 15}          & \cellcolor{red!20}S & \cellcolor{green!20}H & \cellcolor{green!20}H & \cellcolor{green!20}H & \cellcolor{green!20}H & \cellcolor{green!20}H & \cellcolor{green!20}H  & \cellcolor{green!20}H  & \cellcolor{green!20}H  & \cellcolor{green!20}H  & \cellcolor{green!20}H  & \cellcolor{green!20}H  & \cellcolor{green!20}H  & \cellcolor{green!20}H  & \cellcolor{green!20}H  & \cellcolor{green!20}H  & \cellcolor{green!20}H  \\ \hline
		\cellcolor{RoyalBlue!20} \textbf{Soft 16}          & \cellcolor{red!20}S & \cellcolor{red!20}S & \cellcolor{red!20}S & \cellcolor{green!20}H & \cellcolor{green!20}H & \cellcolor{green!20}H & \cellcolor{green!20}H  & \cellcolor{green!20}H  & \cellcolor{green!20}H  & \cellcolor{green!20}H  & \cellcolor{green!20}H  & \cellcolor{green!20}H  & \cellcolor{green!20}H  & \cellcolor{green!20}H  & \cellcolor{green!20}H  & \cellcolor{green!20}H  & \cellcolor{green!20}H  \\ \hline
		\cellcolor{RoyalBlue!20} \textbf{Soft 17 - 20}     & \cellcolor{red!20}S & \cellcolor{red!20}S & \cellcolor{red!20}S & \cellcolor{red!20}S & \cellcolor{green!20}H & \cellcolor{green!20}H & \cellcolor{green!20}H  & \cellcolor{red!20}S & \cellcolor{red!20}S  & \cellcolor{red!20}S  & \cellcolor{red!20}S & \cellcolor{green!20}H  & \cellcolor{green!20}H  & \cellcolor{red!20}S  & \cellcolor{green!20}H  & \cellcolor{green!20}H  & \cellcolor{green!20}H  \\ \hline
	\end{tabular}\caption{Basic strategy for the two up-card variation (single deck) \textit{for hard dealer totals only}.  The $^*$ signifies that the optimal decision depends on the specific layout given this hand total.}
	\label{table:BS2UCHard}
\end{table}

\newpage

\begin{table}[h]
	\setlength{\arrayrulewidth}{0.5mm}
	\begin{tabular}{|l|c|c|c|c|c|c|c|c|c|}
		\hline
		\rowcolor{RoyalBlue!20} \backslashbox{\textbf{\textbf{Player}}}{\textbf{Dealer}}& \textbf{12} & \textbf{13} & \textbf{14} & \textbf{15} & \textbf{16} & \textbf{17} & \textbf{18} & \textbf{19} & \textbf{20} \\ \hline
		\cellcolor{RoyalBlue!20}\textbf{Hard 4 - 13}      & \cellcolor{green!20}H           & \cellcolor{green!20}H           & \cellcolor{green!20}H           & \cellcolor{green!20}H           & \cellcolor{green!20}H           & \cellcolor{green!20}H           & \cellcolor{green!20}H           & \cellcolor{green!20}H           & \cellcolor{green!20}H           \\ \hline
		\cellcolor{RoyalBlue!20}\textbf{Hard 14}          & \cellcolor{red!20}S           & \cellcolor{red!20}S           & \cellcolor{green!20}H           & \cellcolor{green!20}H           & \cellcolor{green!20}H           & \cellcolor{green!20}H           & \cellcolor{green!20}H           & \cellcolor{green!20}H           & \cellcolor{green!20}H           \\ \hline
		\cellcolor{RoyalBlue!20}\textbf{Hard 15}          & \cellcolor{red!20}S           & \cellcolor{red!20}S           & \cellcolor{red!20}S           & \cellcolor{red!20}S           & \cellcolor{green!20}H           & \cellcolor{green!20}H           & \cellcolor{green!20}H           & \cellcolor{green!20}H           & \cellcolor{green!20}H           \\ \hline
		\cellcolor{RoyalBlue!20}\textbf{Hard 16}          & \cellcolor{red!20}S           & \cellcolor{red!20}S           & \cellcolor{red!20}S           & \cellcolor{red!20}S           & \cellcolor{red!20}S           & \cellcolor{red!20}S*          & \cellcolor{green!20}H           & \cellcolor{green!20}H           & \cellcolor{green!20}H           \\ \hline
		\cellcolor{RoyalBlue!20}\textbf{Hard 17}          & \cellcolor{red!20}S           & \cellcolor{red!20}S           & \cellcolor{red!20}S           & \cellcolor{red!20}S           & \cellcolor{red!20}S           & \cellcolor{red!20}S           & \cellcolor{green!20}H           & \cellcolor{green!20}H           & \cellcolor{green!20}H           \\ \hline
		\cellcolor{RoyalBlue!20}\textbf{Hard and Soft 18} & \cellcolor{red!20}S           & \cellcolor{red!20}S           & \cellcolor{red!20}S           & \cellcolor{red!20}S           & \cellcolor{red!20}S           & \cellcolor{red!20}S           & \cellcolor{red!20}S           & \cellcolor{green!20}H           & \cellcolor{green!20}H           \\ \hline
		\cellcolor{RoyalBlue!20}\textbf{Hard and Soft 19} & \cellcolor{red!20}S           & \cellcolor{red!20}S           & \cellcolor{red!20}S           & \cellcolor{red!20}S           & \cellcolor{red!20}S           & \cellcolor{red!20}S           & \cellcolor{red!20}S           & \cellcolor{red!20}S           & \cellcolor{green!20}H           \\ \hline
		\cellcolor{RoyalBlue!20}\textbf{Hard and Soft 20} & \cellcolor{red!20}S           & \cellcolor{red!20}S           & \cellcolor{red!20}S           & \cellcolor{red!20}S           & \cellcolor{red!20}S           & \cellcolor{red!20}S           & \cellcolor{red!20}S           & \cellcolor{red!20}S           & \cellcolor{red!20}S           \\ \hline
		\cellcolor{RoyalBlue!20}\textbf{Soft 12 - 15}     & \cellcolor{green!20}H           & \cellcolor{green!20}H           & \cellcolor{green!20}H           & \cellcolor{green!20}H           & \cellcolor{green!20}H           & \cellcolor{green!20}H           & \cellcolor{green!20}H           & \cellcolor{green!20}H           & \cellcolor{green!20}H           \\ \hline
		\cellcolor{RoyalBlue!20}\textbf{Soft 16}          & \cellcolor{green!20}H           & \cellcolor{red!20}S           & \cellcolor{red!20}S           & \cellcolor{red!20}S           & \cellcolor{green!20}H           & \cellcolor{green!20}H           & \cellcolor{green!20}H           & \cellcolor{green!20}H           & \cellcolor{green!20}H           \\ \hline
		\cellcolor{RoyalBlue!20}\textbf{Soft 17}          & \cellcolor{red!20}S           & \cellcolor{red!20}S           & \cellcolor{red!20}S           & \cellcolor{red!20}S           & \cellcolor{red!20}S           & \cellcolor{red!20}S           & \cellcolor{green!20}H           & \cellcolor{green!20}H           & \cellcolor{green!20}H           \\ \hline
	\end{tabular}\caption{Basic strategy for the two up-card variation (single deck) \textit{for soft dealer totals only}.  The $^*$ signifies that the optimal decision depends on the specific layout given this hand total.}
\label{table:BS2UCSoft}
\end{table}

\subsection{One Up-Card}
The following table give the basic strategy for the one up-card variation played using one deck.  

\begin{table}[h]
	\setlength{\arrayrulewidth}{0.5mm}
	\begin{tabular}{|l|c|c|c|c|c|c|c|c|c|c|}
		\hline
		\rowcolor{RoyalBlue!20}
		\backslashbox{\textbf{\textbf{Player}}}{\textbf{Dealer}} & \textbf{2} & \textbf{3} & \textbf{4} & \textbf{5} & \textbf{6} & \textbf{7} & \textbf{8} & \textbf{9} & \textbf{10} & \textbf{11} \\ \hline
		\cellcolor{RoyalBlue!20} \textbf{Hard 4 - 12}   & \cellcolor{green!20}H        & \cellcolor{green!20}H        & \cellcolor{green!20}H       & \cellcolor{green!20}H       & \cellcolor{green!20}H       & \cellcolor{green!20}H        & \cellcolor{green!20}H        & \cellcolor{green!20}H        & \cellcolor{green!20}H        & \cellcolor{green!20}H         \\ \hline
		\cellcolor{RoyalBlue!20} \textbf{Hard 13}                 & \cellcolor{green!20}H        & \cellcolor{green!20}H        &  \cellcolor{green!20}H       & \cellcolor{green!20}H       & \cellcolor{green!20}H*       & \cellcolor{green!20}H       & \cellcolor{green!20}H        & \cellcolor{green!20}H        & \cellcolor{green!20}H         & \cellcolor{green!20}H         \\ \hline
		\cellcolor{RoyalBlue!20} \textbf{Hard 14}                 & \cellcolor{red!20}S*     & \cellcolor{green!20}H*       & \cellcolor{green!20}H*       & \cellcolor{red!20}S      & \cellcolor{red!20}S      & \cellcolor{green!20}H        & \cellcolor{green!20}H        & \cellcolor{green!20}H        & \cellcolor{green!20}H         & \cellcolor{green!20}H         \\ \hline
		\cellcolor{RoyalBlue!20} \textbf{Hard 15}                 & \cellcolor{red!20}S      & \cellcolor{red!20}S      & \cellcolor{red!20}S      & \cellcolor{red!20}S      & \cellcolor{red!20}S      & \cellcolor{green!20}H        & \cellcolor{green!20}H        & \cellcolor{green!20}H        & \cellcolor{green!20}H         & \cellcolor{green!20}H         \\ \hline
		\cellcolor{RoyalBlue!20} \textbf{Hard 16}                 & \cellcolor{red!20}S      & \cellcolor{red!20}S      & \cellcolor{red!20}S      & \cellcolor{red!20}S      & \cellcolor{red!20}S      & \cellcolor{green!20}H      & \cellcolor{green!20}H      & \cellcolor{green!20}H*       & \cellcolor{green!20}H*        & \cellcolor{red!20}S       \\ \hline
		\cellcolor{RoyalBlue!20} \textbf{Hard 17 - 20} & \cellcolor{red!20}S      & \cellcolor{red!20}S      & \cellcolor{red!20}S      & \cellcolor{red!20}S      & \cellcolor{red!20}S      & \cellcolor{red!20}S      & \cellcolor{red!20}S      & \cellcolor{red!20}S      & \cellcolor{red!20}S       & \cellcolor{red!20}S       \\ \hline
		\cellcolor{RoyalBlue!20} \textbf{Soft 12 - 16}                & \cellcolor{green!20}H        & \cellcolor{green!20}H        & \cellcolor{green!20}H        & \cellcolor{green!20}H        & \cellcolor{green!20}H        & \cellcolor{green!20}H        & \cellcolor{green!20}H        & \cellcolor{green!20}H        & \cellcolor{green!20}H         & \cellcolor{green!20}H         \\ \hline
		\cellcolor{RoyalBlue!20} \textbf{Soft 17}                & \cellcolor{red!20}S        & \cellcolor{red!20}S        & \cellcolor{red!20}S        & \cellcolor{green!20}H        & \cellcolor{green!20}H        & \cellcolor{green!20}H        & \cellcolor{green!20}H        & \cellcolor{green!20}H        & \cellcolor{green!20}H         & \cellcolor{green!20}H         \\ \hline
		\cellcolor{RoyalBlue!20} \textbf{Soft 18}  & \cellcolor{red!20}S      & \cellcolor{red!20}S      & \cellcolor{red!20}S      & \cellcolor{red!20}S      & \cellcolor{red!20}S      & \cellcolor{red!20}S      & \cellcolor{red!20}S      & \cellcolor{green!20}H      & \cellcolor{red!20}S       & \cellcolor{red!20}S       \\ \hline 
		\cellcolor{RoyalBlue!20} \textbf{Soft 18 - 20}    & \cellcolor{red!20}S      & \cellcolor{red!20}S      & \cellcolor{red!20}S      & \cellcolor{red!20}S      & \cellcolor{red!20}S      & \cellcolor{red!20}S      & \cellcolor{red!20}S      & \cellcolor{red!20}S      & \cellcolor{red!20}S       & \cellcolor{red!20}S       \\ \hline
	\end{tabular}
	\caption{Basic strategy for the one up-card variation (single deck).  The $^*$ signifies that the optimal decision depends on the specific layout given this hand total.}
	\label{table:BS1UC}
\end{table}

\subsection{No Up-Card}
The following table give the basic strategy for the no up-card variation played using one deck. 
\begin{table}[h] 
	\setlength{\arrayrulewidth}{0.5mm}
	\begin{tabular}{|l|c|}
		\hline
		\rowcolor{RoyalBlue!20}
		\textbf{Player Total} & \textbf{Decision} \\ \hline
		\cellcolor{RoyalBlue!20}\textbf{Hard 2 - 14}         &  \cellcolor{green!20}H              \\ \hline
		\cellcolor{RoyalBlue!20}\textbf{Hard 15 - 20}        & \cellcolor{red!20}S             \\ \hline
		\cellcolor{RoyalBlue!20}\textbf{Soft 12 - 17}        & \cellcolor{green!20}H              \\ \hline
		\cellcolor{RoyalBlue!20}\textbf{Soft 18 - 20}        & \cellcolor{red!20}S             \\ \hline
	\end{tabular}
	\caption{Basic strategy for the no up-card variation when using one deck with, and without, replacement.}
	\label{table:BSNUC}
\end{table}

\begin{remark}
	In Tables \ref{table:BS2UCHard} and \ref{table:BS2UCSoft}, the only changes to the basic strategy when we play using one deck with replacement are the following:
	\begin{itemize}
		\item Hard 15 v. Hard 10 becomes hit.
		\item Soft 17 v. Hard 10 becomes stand.
		\item Soft 16 v. Soft 13 becomes hit.
	\end{itemize}
	In Table \ref{table:BS1UC}, the following decisions change when we play with replacement:
	\begin{itemize}
		\item Hard 16 v. 7 and Hard 16 v. 10 both become stand.
		\item Soft 17 v. 5 becomes stand.
		\item Soft 18 v. 9 becomes stand.
	\end{itemize}
	Lastly, in Table \ref{table:BSNUC}, there are no changes to the basic strategy when we play with replacement. Additionally, all asterisks in the above tables are removed when playing with replacement.
\end{remark}

\printbibliography

@Article{Ethier2019,
  author        = {Ethier, Stewart N. and Lee, Jiyeon},
  title         = {Snackjack: A toy model of blackjack},
  year          = {2019},
  month         = jun,
  archiveprefix = {arXiv},
  copyright     = {arXiv.org perpetual, non-exclusive license},
  doi           = {10.48550/ARXIV.1906.01220},
  eprint        = {1906.01220},
  primaryclass  = {math.PR},
  publisher     = {arXiv},
}

@Book{Thorp1966,
  author    = {Edward O. Thorp},
  publisher = {Knopf Doubleday Publishing Group},
  title     = {Beat the Dealer},
  year      = {1966},
}

@Article{Baldwin1956,
  author   = {Baldwin, Roger R. and Cantey, Wilbert E. and Maisel, Herbert and McDermott, James P.},
  journal  = {Journal of the American Statistical Association},
  title    = {The optimum strategy in blackjack},
  year     = {1956},
  issn     = {0162-1459},
  pages    = {429--439},
  volume   = {51},
  doi      = {10.2307/2281431},
  language = {English},
  zbl      = {0071.35202},
  zbmath   = {3119643},
}

@Article{Nairn2019,
  author        = {Nairn, John A.},
  title         = {Exact Calculation of Expected Values for Splitting Pairs in Blackjack},
  year          = {2019},
  month         = sep,
  archiveprefix = {arXiv},
  copyright     = {arXiv.org perpetual, non-exclusive license},
  doi           = {10.48550/ARXIV.1909.13710},
  eprint        = {1909.13710},
  primaryclass  = {cs.OH},
  publisher     = {arXiv},
}

@Article{Mura2020,
  author        = {Mura, Yushi and Wada, Hiroki},
  journal       = {Progress of Theoretical and Experimental Physics},
  title         = {Quantization of Blackjack: Quantum Basic Strategy and Advantage},
  year          = {2020},
  issn          = {2050-3911},
  month         = oct,
  number        = {10},
  volume        = {2021},
  archiveprefix = {arXiv},
  copyright     = {arXiv.org perpetual, non-exclusive license},
  date          = {2020-11-24},
  doi           = {10.1093/ptep/ptab125},
  eprint        = {2011.12342},
  primaryclass  = {quant-ph},
  publisher     = {Oxford University Press (OUP)},
}

@Article{Lin2019,
  author        = {Lin, Joseph X. and Formaggio, Joseph A. and Harrow, Aram W. and Natarajan, Anand V.},
  journal       = {Phys. Rev. A 102, 012425 (2020)},
  title         = {Quantum Blackjack or Can MIT Bring Down the House Again?},
  year          = {2019},
  issn          = {2469-9934},
  month         = jul,
  number        = {1},
  pages         = {012425},
  volume        = {102},
  archiveprefix = {arXiv},
  copyright     = {arXiv.org perpetual, non-exclusive license},
  date          = {2019-08-26},
  doi           = {10.1103/physreva.102.012425},
  eprint        = {1908.09417},
  primaryclass  = {quant-ph},
  publisher     = {American Physical Society (APS)},
}

@Book{Epstein1967,
  author    = {R. A. Epstein},
  publisher = {Academic Press, New York},
  title     = {The Theory of Gambling and Statistical Logic},
  year      = {1967},
}
\end{document}